\documentclass[conference]{IEEEtran}
\IEEEoverridecommandlockouts
\usepackage{microtype}

\usepackage{mathrsfs}
\usepackage{amsfonts}
\usepackage{amsmath}
\usepackage{amsthm}
\usepackage{amssymb}
\usepackage{dsfont}
\usepackage{esint}
\usepackage{bbm}
\usepackage{bm}
\usepackage{mathtools}

\usepackage{enumerate}
\usepackage[shortlabels]{enumitem}
\usepackage{framed}
\usepackage{csquotes}

\usepackage{float}
\usepackage{tabularx}
\usepackage{xcolor}
\usepackage{multicol}
\usepackage{subcaption}
\usepackage{caption}
\usepackage{cite}

\usepackage{hyperref}

\usepackage{listings}
\allowdisplaybreaks

\usepackage{titlesec}
\usepackage[many]{tcolorbox}
\usepackage[dvipsnames]{xcolor}

\newtheorem{theorem}{Theorem}

\newtheorem{lemma}[theorem]{Lemma}
\newtheorem{remark}[theorem]{Remark}
\newtheorem{proposition}[theorem]{Proposition}
\newtheorem{definition}[theorem]{Definition}

\tcbuselibrary{skins, breakable}
\newcommand{\RR}{\mathbb{R}}

\newcommand{\ZZ}{\mathbb{Z}} 
\newcommand{\ii}{\mathbbm{1}} 
\newcommand{\ket}[1]{|#1\rangle} 
\newcommand{\bra}[1]{\langle #1|} 
\newcommand{\op}[2]{|#1\rangle\langle #2|}

\newcommand{\Tr}{\operatorname{Tr}} 
 
\newcommand{\mc}[1]{\mathcal{#1}} 
\newcommand{\fall}{\textrm{ }\forall\textrm{ }} 
\newcommand{\norm}[1]{\left\lVert#1\right\rVert} 
\newcommand{\abs}[1]{\left|#1\right|} 
\newcommand{\ra}{\rightarrow}

\newcommand{\ep}{\varepsilon} 
\newcommand{\bp}[1]{\!\left(#1\right)} 
\newcommand{\bb}[1]{\!\left[#1\right]}
\newcommand{\bc}[1]{\!\left\{#1\right\}} 
 
\newcommand{\hilb}{\mc{H}}

\def\BibTeX{{\rm B\kern-.05em{\sc i\kern-.025em b}\kern-.08em
    T\kern-.1667em\lower.7ex\hbox{E}\kern-.125emX}}

\allowdisplaybreaks

\begin{document}

% to remove for ISIT submission, but keep for arxiv in order to have page numbering
\thispagestyle{plain}
\pagestyle{plain}

\title{Genuine multipartite Rains entanglement\\

\thanks{}
}

\author{%
  \IEEEauthorblockN{Hailey S. Murray\IEEEauthorrefmark{1}\IEEEauthorrefmark{2}, Sagnik Bhattacharya\IEEEauthorrefmark{3}, M. Cerezo\IEEEauthorrefmark{4}\IEEEauthorrefmark{2}, Liuke Lyu\IEEEauthorrefmark{5} and Mark M.~Wilde\IEEEauthorrefmark{6}\IEEEauthorrefmark{2}}
  \IEEEauthorblockA{\IEEEauthorrefmark{1} School of Applied and Engineering Physics, Cornell University,\\Ithaca, NY 14850, USA. Email: hm649@cornell.edu}
  \IEEEauthorblockA{\IEEEauthorrefmark{2}Quantum Science Center, Oak Ridge, TN 37931, USA.}
  \IEEEauthorblockA{\IEEEauthorrefmark{3} Department of Electrical and Computer Engineering, University of Maryland, College Park, MD 20742, USA.}
  \IEEEauthorblockA{\IEEEauthorrefmark{4}Information Sciences, Los Alamos National Laboratory, NM 87545, USA.}
  \IEEEauthorblockA{\IEEEauthorrefmark{5}D\'epartment de Physique, Universit\'e de Montr\'eal, Montr\'eal, QC H3C 3J7, Canada.}
  \IEEEauthorblockA{\IEEEauthorrefmark{6}School of Electrical and Computer
  Engineering,\\Cornell University, Ithaca, New York 14850, USA.
                    Email: wilde@cornell.edu 
                    }
                    \thanks{H.~S.~M.~acknowledges support from NSF grant No.~DGE-2139899. M.~C.~acknowledges support by the Laboratory Directed Research and Development (LDRD) program of Los Alamos National Laboratory (LANL) under project number 20260043DR.
This work was supported by the Quantum Science Center (QSC), a National Quantum Information Science Research Center of the U.S. Department of Energy (DOE).}
}

\maketitle

\begin{abstract}We introduce the genuine multipartite Rains entanglement (GMRE) as a measure of genuine multipartite entanglement that can be computed using semi-definite programming. Similar to the Rains relative entropy (its bipartite counterpart), the GMRE is monotone under selective quantum operations that completely preserve the positivity of the partial transpose, implying that it is a multipartite entanglement monotone. As a consequence, we show that the GMRE bounds both the one-shot standard and probabilistic approximate GHZ-distillable entanglement from above. We also develop a generalization of this quantity that incorporates other entropies, including quantum R\'enyi relative entropies. 

\end{abstract}

\begin{IEEEkeywords}
Quantum entanglement, multipartite systems, multipartite entanglement measures, distillable entanglement.
\end{IEEEkeywords}

% ----====----====----==== introduction ====----====----====----
\section{Introduction}

% ----====----====----==== background/motivation ====----====----====----
\textbf{Background and motivation}:
Entanglement is a distinguishing feature of quantum mechanics, with pure maximally entangled states being vital components in foundational protocols~\cite{bennet_etal1993,bennett_wiesner1992}. As such, distilling pure maximally entangled states is a pertinent task. However, distillable entanglement~\cite{bennett_divincenzo1996} is not convex~\cite{shor_smolin_terhal2001} and may not even be computable~\cite{wolf_cubitt_perezgarcia2011}, motivating the search for bounds on it.

The Rains relative entropy  is an entanglement measure that is efficiently computable using semi-definite programming (SDP)~\cite{rains2001} and does not increase on average under selective quantum operations that completely preserve the positivity of the partial transpose~\cite{eisert2022}---a property implying that the quantity is an entanglement monotone~\cite{vidal2000}. The Rains relative entropy is also an upper bound on the bipartite distillable entanglement~\cite{rains2001} and the probabilistic approximate distillable entanglement (PADE)~\cite{eisert2022}, the latter property being stronger than the former due to the PADE being no smaller than the standard distillable entanglement~\cite{eisert2022}.

Characterizing the entanglement of multipartite systems has broad applications, including many-body systems~\cite{wang_etal2025} and quantum cryptography~\cite{cleve_gottesman_lo1999}, but it is a more complex task than its bipartite counterpart. In this setting, one goal is to detect ``genuine multipartite entanglement'' (GME)~\cite{Dur2000,Seevinck2001} (see also~\cite{ma_li_shang2025}). Quantum states that exhibit GME cannot be written as a convex combination of states for which each state in the decomposition is separable with respect to an arbitrary bipartition of the system.

% ----====----====----==== summary of contributions ====----====----====----
\textbf{Summary of contributions}: Our paper introduces the genuine multipartite Rains entanglement (GMRE) as a genuine multipartite entanglement measure that generalizes the Rains relative entropy. After covering preliminary topics, we define the GMRE, as well as a generalization that allows for  incorporating other quantum entropies. We show that the GMRE can be written as a convex optimization problem and can be computed using semi-definite programming. We then prove that the GMRE satisfies several desirable properties, primarily focusing on selective positive partial transpose (PPT) monotonicity. Next, we provide a lower bound for the GMRE based on approximate GHZ states and use this to establish upper bounds on the one-shot GHZ-distillable entanglement and PADE. We also prove that a regularized GMRE measure bounds the asymptotic GHZ-distillable entanglement and PADE, with a single-letter upper bound being out of reach due to difficulties proving subadditivity of GMRE, in turn related to the lack of tensor stability of the set of biseparable states~\cite{Navascues2020}. Finally, we write the set over which the GMRE is minimized in terms of semi-definite constraints.

% ----====----====----==== preliminaries ====----====----====----
\section{Preliminaries}

Let $\mc{L}(\hilb)$ denote the set of linear operators acting on a Hilbert space $\mc{H}$,  $\mc{L}_+(\hilb)$ the set of positive semi-definite (PSD) operators, and $\mc{S}(\mc{H})$ the set of quantum states (PSD operators with unit trace). Let $\operatorname{CP}(\hilb,\hilb')$ denote the set of completely positive maps from $\mc{L}(\hilb)$ to $\mc{L}(\hilb')$, $\textrm{CPTP}(\hilb,\hilb')$ the set of quantum channels (completely positive trace-preserving maps), and $\textrm{LOCC}(\hilb,\hilb')$ the set of LOCC quantum channels~\cite{chitambar2014}. In general, we will use $\hilb_k$ to refer to the tensor-product space $\hilb_{A_1}\otimes\hilb_{A_2}\otimes\ldots\otimes\hilb_{A_k}$ where $k\in\ZZ^{\geq2}$.

The fidelity of states $\rho,\sigma\in\mc{S}(\hilb)$ is defined as~\cite{uhlmann1976}
\begin{align}\label{def:fidelity}
    F(\rho,\sigma)\coloneq\norm{\sqrt{\rho}\sqrt{\sigma}}_1^2.
\end{align}
where $\norm{\cdot}_1\coloneq\Tr[\abs{\cdot}]$ is the trace norm. The binary entropy of a probability distribution $(p,1-p)$ is defined for $p\in[0,1]$ as
\begin{align}
    h_2(p)\coloneq-p\log_2p-(1-p)\log_2(1-p).\label{def:binary_entropy}
\end{align}
The quantum relative entropy of $\rho\in \mc{S}(\hilb)$ and $\sigma\in \mc{L_+}(\hilb)$ is defined as~\cite{umegaki1962}
\begin{equation}
    D(\rho\|\sigma)\coloneq
    \begin{cases}
        \Tr[\rho(\log_2\rho-\log_2\sigma)] & \textrm{if }(\star) \\
        +\infty & \textrm{else}
    \end{cases}
    \label{def:quantum_rel_entropy}
\end{equation}
where $(\star)\equiv\textrm{supp}(\rho)\subseteq\textrm{supp}(\sigma)$. The sandwiched R\'enyi relative entropy $\widetilde{D}_\alpha$ is defined for $\rho\in\mc{S}(\hilb)$, $\sigma\in\mc{L}_+(\hilb)$, and $\alpha\in(0,1)\cup(1,\infty)$ as~\cite{mullerlennert2013,wilde2014}
\begin{equation}
    \widetilde{D}_\alpha(\rho\|\sigma)\coloneq
    \begin{cases}
        \frac{1}{\alpha-1}\log_2\Tr\bb{\bp{\sigma^{\frac{1-\alpha}{2\alpha}}\rho\sigma^{\frac{1-\alpha}{2\alpha}}}^\alpha} & \textrm{if }(\star\star) \\
        +\infty & \textrm{else,}
    \end{cases}
    \label{def:sandwiched_quantum_rel_entropy}
\end{equation}
where $(\star\star)\equiv \alpha\in(0,1) \lor \left(\alpha\in(1,\infty)  \land (\star)\right)$,
and it satisfies the following~\cite{mullerlennert2013,wilde2014}:
\begin{align}
    \lim_{\alpha\ra1}\widetilde{D}_\alpha(\rho\|\sigma)&=D(\rho\|\sigma).
    \label{res:limit_sandwich_renyi_quantum_rel_entropy}
\end{align}
Given $\mc{N}\in\textrm{CPTP}(\mc{H},\mc{H}')$, for all $\rho\in \mc{S}(\mc{H})$, $\sigma\in\mc{L}_+(\mc{H})$, and $\alpha\in[\frac{1}{2},1)\cup(1,\infty)$, the data-processing inequality asserts that
\begin{align}
    D(\rho\|\sigma)&\geq D(\mc{N}(\rho)\|\mc{N}(\sigma)), \label{res:data_processing_ineq}\\
    \widetilde{D}_\alpha(\rho\|\sigma)&\geq\widetilde{D}_\alpha(\mc{N}(\rho)\|\mc{N}(\sigma)),
    \label{res:data_processing_sandwich_renyi}
\end{align}
where~\eqref{res:data_processing_ineq} was established by~\cite{lindblad1975} and~\eqref{res:data_processing_sandwich_renyi} by~\cite{frank2013} (see also~\cite{wilde2018}). Let $\mc{X}$ be a finite alphabet, $p\colon\mc{X}\ra[0,1]$ a probability distribution, and $q\colon\mc{X}\ra[0,\infty)$ a non-negative function. Consider a classical-quantum state $\omega\in \mc{S}(\hilb_{XA})$ and operator $\tau\in\mc{L}_+(\hilb_{XA})$ of the form
\begin{align}
    &\omega\coloneq\sum_xp(x)\op{x}{x}\otimes\omega_x,&\tau\coloneq\sum_xq(x)\op{x}{x}\otimes\tau_x,&
    \label{def:states_in_direct_sum_prop}
\end{align}
where $\omega_x\in \mc{S}(\hilb_A)$ and $\tau_x\in\mc{L}_+(\hilb_A)$ for all $x$. The quantum relative entropy satisfies a direct-sum equality~\cite{khatri2024} and the sandwiched R\'enyi relative entropy a direct-sum inequality~\cite{eisert2022} for $\alpha>1$:
\begin{align}
    D(\omega\|\tau)&=D(p\| q)+\sum_xp(x)D(\omega_x\|\tau_x), \label{direct_sum_prop} \\
    \widetilde{D}_\alpha(\omega\|\tau)&\geq D(p\|q)+\sum_xp(x)\widetilde{D}_\alpha(\omega_x\|\tau_x),
    \label{res:direct_sum_sandwich_renyi}
\end{align}
where $D(p\| q)$ is the classical relative entropy of probability distributions $p$ and $q$, defined as
\begin{equation}
    D(p\|q)\coloneq\begin{cases}
        \sum_xp(x)\log_2\bp{\frac{p(x)}{q(x)}}&\textrm{if }(\star\star\star)\\
        +\infty&\textrm{else,}
    \end{cases}
    \label{classical_rel_entropy}
\end{equation}
where $(\star\star\star)\equiv \textrm{supp}(p)\subseteq\textrm{supp}(q)$. For $\rho\in\mc{S}(\hilb_{AB})$, the Rains relative entropy is defined as~\cite{rains2001}
\begin{align}
    R(\rho)\coloneq\inf_{\sigma\in\mc{S}(\hilb_{AB})}\bc{D(\rho\|\sigma)+\log_2\norm{T_B(\sigma)}_1},
\end{align}
where $T_B(\cdot)\coloneq\sum_{i,j}(\ii^d\otimes\op{i}{j})(\cdot)(\ii^d\otimes\op{i}{j})$ is the partial transpose with respect to subsystem $B$, $\{\ket{i}\}_{i=0}^{d-1}$ is the computational basis of $\mc{H}_B$, and $\ii^d$ is the $d\times d$ identity matrix. A quantum state is positive partial transpose (PPT) if it is positive semi-definite under the action of the partial transpose~\cite{peres1996, horodecki1996}. 

A multipartite state is genuinely entangled if it is not biseparable, meaning that it cannot be written as a convex combination of states that are separable with respect to a bipartition. For example, a tripartite state $\rho\in\mc{S}(\hilb_{ABC})$ is \textit{not} genuinely entangled if it can be written as $\rho=p_{A}\rho_{A|BC}+p_{B}\rho_{B|AC}+p_{C}\rho_{C|AB}$, where $(p_m)_m$ is a probability distribution and each $\rho_{m}$ is separable with respect to the bipartition $m$. Since directly checking for separability is non-trivial, we will instead focus on PPT mixtures (i.e., mixtures of states that are PPT with respect to any bipartition). Such states are easier to characterize using the PPT criterion, which is an operational necessary condition for separability~\cite{peres1996, horodecki1996}.

% ----====----====----==== GMRE definition ====----====----====----
\section{GMRE definition}

We define the GMRE of $\rho\in\mc{S}(\hilb_k)$~as
\begin{align}
    R(\rho)&\coloneq\hspace{-.3cm}\inf_{\substack{\sigma\in\mc{S}(\hilb_{k}), \\ \sigma=\sum_mr_m\omega_m}}\bc{\!D\bp{\rho\middle\|\sigma}+\log_2\bp{\sum_mr_m\norm{T_m(\omega_m)}_1}\!}\!,\nonumber
\end{align}
where $m$ labels a unique bipartition of the system, $(r_m)_m$ is a probability distribution, and $\omega_m\in\mc{S}(\hilb_k)$ for each $m$. Note that, for $k=2$, the GMRE reduces to the bipartite Rains relative entropy, as there is only one unique bipartition of a bipartite system. By setting $\sigma = \rho$, the GMRE reduces to a genuine multipartite generalization of the log-negativity~\cite{Zyczkowski1998VolumeSeparability,VidalWerner2002Negativity,plenio2005}, the latter being a well known bipartite entanglement measure. The resultant formula is related to the genuine multipartite negativity, introduced in~\cite{jungnitsch2011}, in the same way that the bipartite log-negativity and negativity are related.

Motivated by the approach in~\cite{eisert2022}, we now extend the GMRE definition to incorporate more general entropies. For a function $\boldsymbol{D}\colon\mc{S}(\hilb_k)\times\mc{L}_+(\hilb_k)\ra\RR$ satisfying the data-processing inequality~\eqref{res:data_processing_ineq} and the scaling property, $\boldsymbol{D}(\rho\|c\sigma)=\boldsymbol{D}(\rho\|\sigma)-\log_2c$ for $c\in\RR^+$, we define the multipartite $\boldsymbol{D}$-Rains entanglement of $\rho\in\mc{S}(\hilb_k)$ as
\begin{align}
    \boldsymbol{R}(\rho)&\coloneq\hspace{-.3cm}\inf_{\substack{\sigma\in\mc{S}(\hilb_{k}), \\ \sigma=\sum_mr_m\omega_m}}\bc{\!\boldsymbol{D}\bp{\rho\middle\|\sigma}+\log_2\bp{\sum_mr_m\norm{T_m(\omega_m)}_1}\!}\!.\nonumber
\end{align}
Since $\widetilde{D}_\alpha$ in \eqref{def:sandwiched_quantum_rel_entropy} satisfies the scaling property and the data-processing inequality for $\alpha\in[\frac{1}{2},1)\cup(1,\infty)$, we define a $\boldsymbol{D}$-Rains entanglement measure with $\boldsymbol{D}\ra\widetilde{D}_\alpha$, which we call the sandwiched R\'enyi--Rains entanglement and denote by $\widetilde{R}_\alpha$.

Let us now show that the GMRE can be viewed as a convex optimization problem. Define a set
    \begin{multline}
        \mc{T}(\hilb_{k})\coloneq\biggl\{\tau=\sum_mq_m\tau_m\in\mc{L}^+(\hilb_{k}):q_m,\tau_m\geq0\fall m,\\
        \sum_mq_m\norm{T_m(\tau_m)}_1\leq1\biggr\},\label{def:general_T_set_definition}
    \end{multline}
    where $m$ labels a unique bipartition of the system. Observe that each $q_m$ can be absorbed into the corresponding operator $\tau_m$, so that we can also define $\mc{T}(\hilb_k)$ with $q_m\tau_m\ra\widetilde{\tau}_m$.
    
\begin{theorem}\label{thm:general_convex_opt_theorem}
    The following equality holds for every multipartite $\boldsymbol{D}$-Rains entanglement measure and $\rho\in\mc{S}(\mc{H}_{k})$:
    \begin{align}
        \boldsymbol{R}(\rho)&=\inf_{\tau\in\mc{T}(\hilb_{k})}\boldsymbol{D}(\rho\|\tau).
        \label{res:general_multi_rains_as_convex_opt}
    \end{align}
\end{theorem}
\begin{proof}
    See Appendix~\ref{app:multi_rains_convex}.
\end{proof}
\begin{remark}
    The GMRE can be viewed as a convex optimization problem because the quantum relative entropy is jointly convex~\cite{lieb_ruskai1973} and the set $\mc{T}(\hilb_k)$ is convex. Furthermore, the GMRE can be computed using semi-definite programming, by using the relative-entropy optimization methods of~\cite{he_saunderson_fawzi2025,koßmann_schwonnek2025}. We note that the runtime for this computation is exponential in the number of parties but polynomial in dimension of the system for a fixed number of parties. In Appendix~\ref{app:multi_rains_and_gmn}, we show that there is a strict separation between the GMRE and the genuine multipartite log-negativity.
\end{remark}

Theorem~\ref{thm:general_convex_opt_theorem} implies alternate formulas for the genuine multipartite Rains and sandwiched R\'enyi--Rains measures:
\begin{align}
    R(\rho)&=\inf_{\tau\in\mc{T}}D(\rho\|\tau),&\widetilde{R}_\alpha(\rho)&=\inf_{\tau\in\mc{T}}\widetilde{D}_\alpha(\rho\|\tau). \label{def:convex_multi_rains_and_sandwich_renyi_rains}
\end{align}
See Appendix~\ref{app:sandwiched-RR-to-GRME} for a proof that $\lim_{\alpha\ra1}\widetilde{R}_\alpha(\rho)=R(\rho)$.

% ----====----====----==== GMRE properties ====----====----====----
\section{Selective PPT monotonicity of GMRE}

Using the criteria for a PPT map from~\cite{ishizaka2005}, we can extend the definition of a selective PPT operation from~\cite{eisert2022} to the multipartite setting. 

\begin{definition}
\label{def:selective_ppt_op}
    A tuple $\left(\mc{P}_x\right)_x$ is a selective multipartite PPT operation if 
    \begin{enumerate}
        \item $\mc{P}_x\in\operatorname{CP}(\hilb_{k},\hilb'_{k})$ for all $x$,
        \item $T_{m}\circ\mc{P}_x\circ T_{m}\in\operatorname{CP}(\hilb_{k},\hilb'_{k})$ for all $x$ and for all $m$,     where $m$ labels corresponding bipartitions of $\hilb_k$ and $\hilb'_k$.    
        \item The sum map $\sum_x\mc{P}_x$ is trace preserving.
    \end{enumerate}
\end{definition}
Given $\left(\mc{P}_x\right)_x$ and $\sigma\in \mc{S}(\hilb_{k})$, define states $\sigma_x$ and probabilities $q_x$ as $\sigma_x\coloneq\frac{\mc{P}_x(\sigma)}{q_x}$ and $q_x\coloneq\Tr[\mc{P}_x(\sigma)]$. Consider an arbitrary PPT mixed-state decomposition of $\sigma$ given by $\sigma=\sum_mr_m\omega_m$, where $m$ labels the bipartition. For each $\omega_m$, define states $\sigma_{m,x}$ and probabilities $q_{x|m}$ as $\sigma_{m,x}\coloneq\frac{\mc{P}_x(\omega_m)}{q_{x|m}}$ and $q_{x|m}\coloneq\Tr[\mc{P}_x(\omega_m)]$.

Suppose $\sigma$ is a PPT mixture. We will show that the action of a selective PPT operation on $\sigma$ results in a PPT mixture. We can write $\sigma_x$ as $\sigma_x=\frac{1}{q_x}\mc{P}_x\bp{\sum_mr_m\omega_m}=\frac{1}{q_x}\sum_mr_m\mc{P}_x(\omega_m)$. Noting that the partial transpose is self-inverse, it follows that
\begin{align}
    T_m(\mc{P}_x(\omega_m))&=T_m\circ\mc{P}_x(T_m(T_m(\omega_m))) \notag\\
    &=T_m\circ\mc{P}_x\circ T_m(T_m(\omega_m))\geq0.
    \label{px_omegam_ppt_condition}
\end{align}
By definition, $T_m(\omega_m)\geq0$ and $T_m\circ\mc{P}_x\circ T_m$ is completely positive for all $m$. Thus, the fact that $\sigma_x$ is a convex combination of $\mc{P}_x(\omega_m)$ along with~\eqref{px_omegam_ppt_condition} implies that $\sigma_x$ can be expressed as a PPT mixture. Similarly, we see that
\begin{align}
    T_m(\sigma_{m,x})&=T_m\bp{\frac{\mc{P}_x(\omega_m)}{q_{x|m}}} \notag \\
    &=\frac{1}{q_{x|m}}T_m\circ\mc{P}_x\circ T_m(T_m(\omega_m))\geq0,
\end{align}
where the inequality follows from~\eqref{px_omegam_ppt_condition}, and we conclude that $\sigma_{m,x}$ is PPT with respect to $T_m$. 

For each $\omega_m$, we have that $\sum_xq_{x|m}\norm{T_m(\sigma_{m,x})}_1\leq \norm{T_m(\omega_m)}_1$ using~\cite[Eq.~(8)]{plenio2005} or~\cite[Prop.~2.1]{eisert2006}. With this, we see that
\begin{align}
    \sum_mr_m\norm{T_m(\omega_m)}_1&\geq\sum_mr_m\sum_xq_{x|m}\norm{T_m(\sigma_{m,x})}_1 \notag\\
    &=\sum_{m,x}r_mq_{x|m}\norm{T_m(\sigma_{m,x})}_1.
    \label{normIneq}
\end{align}
\begin{theorem}\label{rains_ppt_monotone_thm}
Let $\rho\in \mc{S}(\hilb_{k})$, and let $\left(\mc{P}_x\right)_x$ be a selective PPT operation. For all $x$, define a state $\rho_x$ and a probability $p_x$ as
\begin{align}
    &\rho_x\coloneq\frac{\mc{P}_x(\rho)}{p_x},&p_x\coloneq\Tr[\mc{P}_x(\rho)].&
\end{align}
Let $\mc{X}^+\coloneq\{x:p_x>0\}$. Then, every multipartite $\boldsymbol{D}$-Rains entanglement measure that satisfies the direct-sum inequality~\eqref{res:direct_sum_sandwich_renyi} is a selective PPT monotone. That is, for all $\rho\in\mc{S}(\hilb_k)$,
\begin{align}
    \boldsymbol{R}(\rho)&\geq\sum_{x\in\mc{X}^+}p_x\boldsymbol{R}(\rho_x).
\end{align}
\end{theorem}

\begin{proof}
    We will follow the proof of~\cite[Theorem~6]{eisert2022} to prove this statement. Consider $\sigma\in \mc{S}(\mc{H}_{k})$ with an arbitrary mixed-state decomposition $\sigma=\sum_mr_m\omega_m$, where $m$ labels the bipartition. Define $\sigma_x,q_x,\sigma_{m,x},q_{x|m}$ similarly as after Definition~\ref{def:selective_ppt_op}. Define a quantum channel  $\mc{P}(\cdot)\coloneq\sum_x\op{x}{x}\otimes\mc{P}_x(\cdot)$. Then,
\begin{align}
    \boldsymbol{D}(\rho\|\sigma)&\geq \boldsymbol{D}(\mc{P}(\rho)\|\mc{P}(\sigma)) \notag\\
    &=\boldsymbol{D}\!\left(\sum_x\op{x}{x}\otimes\mc{P}_x(\rho)\middle\|\sum_x\op{x}{x}\otimes\mc{P}_x(\sigma)\right) \notag\\
    &=\boldsymbol{D}\!\left(\sum_xp_x\op{x}{x}\otimes\rho_x\middle\|\sum_xq_x\op{x}{x}\otimes\sigma_x\right) \notag\\
    &\geq D(p\|q)+\sum_xp_x\boldsymbol{D}(\rho_x\|\sigma_x),\label{intres:rains_selective_ppt_monotone_step}
\end{align}
where we used the data-processing inequality~\eqref{res:data_processing_ineq} and the direct-sum inequality~\eqref{res:direct_sum_sandwich_renyi}. Using~\eqref{intres:rains_selective_ppt_monotone_step} with~\eqref{normIneq} and the monotonicity of $\log_2$, we find that
\begin{multline}
    \boldsymbol{D}(\rho\|\sigma)+\log_2\bp{\sum_mr_m\norm{T_m(\omega_m)}_1}
    \geq D(p\|q) + \\\sum_{x\in\mc{X}^+}p_x\boldsymbol{D}(\rho_x\|\sigma_x)
    +\log_2\bp{\sum_{m,x}r_mq_{x|m}\norm{T_m(\sigma_{m,x})}_1}.\label{ineq5parts}
\end{multline}
We can then employ~\eqref{classical_rel_entropy} along with the monotonicity and concavity of $\log_2$ to find that
\begin{align}
    &D(p\|q)+\log_2\bp{\sum_{m,x}r_mq_{x|m}\norm{T_m(\sigma_{m,x})}_1} \notag\\
    &\geq D(p\|q)+\log_2\bp{\sum_{x\in\mc{X}^+}p_x\sum_{m}\frac{r_mq_{x|m}}{p_x}\norm{T_m(\sigma_{m,x})}_1} \notag\\
    &\geq D(p\|q)+\sum_{x\in\mc{X}^+}p_x\log_2\bp{\sum_{m}\frac{r_mq_{x|m}}{p_x}\norm{T_m(\sigma_{m,x})}_1} \notag\\
    &=\sum_{x\in\mc{X}^+}p_x\bp{\log_2\bp{\frac{p_x}{q_x}}+\log_2\bp{\sum_{m}\frac{r_mq_{x|m}}{p_x}\!\norm{T_m(\sigma_{m,x})}_1\!}\!\!} \notag\\
    &=\sum_{x\in\mc{X}^+}p_x\log_2\bp{\sum_{m}\frac{r_mq_{x|m}}{q_x}\norm{T_m(\sigma_{m,x})}_1}.
    \label{simpIneq}
\end{align}
Let us note that $\sum_m\frac{r_mq_{x|m}}{q_x}\sigma_{m,x}$ is indeed a valid mixed-state decomposition of $\sigma_x$:
\begin{align}
    \sum_m\frac{r_mq_{x|m}}{q_x}\sigma_{m,x}&=\sum_m\frac{r_mq_{x|m}}{q_x}\frac{\mc{P}_x(\omega_m)}{q_{x|m}} \notag\\
    &=\frac{1}{q_x}\sum_mr_m\mc{P}_x(\omega_m)=\sigma_x.
\end{align}
Additionally, $\bp{\frac{r_mq_{x|m}}{q_x}}_m$ is a probability distribution because $r_m,q_{x|m},q_x\geq0$ and $\sum_m\frac{r_mq_{x|m}}{q_x}=\frac{1}{q_x}\sum_mr_mq_{x|m}=\frac{1}{q_x}\Tr[\mc{P}_x(\sigma)]=1$. Then, using~\eqref{simpIneq} with~\eqref{ineq5parts}, we arrive at
\begin{multline}
     \boldsymbol{D}(\rho\|\sigma)+\log_2\bp{\sum_mr_m\norm{T_m(\omega_m)}_1}\geq\sum_{x\in\mc{X}^+}p_x\boldsymbol{D}(\rho_x\|\sigma_x) \\
     +\sum_{x\in\mc{X}^+}p_x\log_2\bp{\sum_{m}\frac{r_mq_{x|m}}{q_x}\norm{T_m(\sigma_{m,x})}_1} \\
     \geq\sum_{x\in\mc{X}^+}p_x\boldsymbol{R}(\rho_x).
     \label{eq:pf-key-ineq}
\end{multline}
Since~\eqref{eq:pf-key-ineq} holds for all $\sigma$, we can take an infimum over $\sigma$ and conclude that $\boldsymbol{R}(\rho)\geq\sum_{x\in\mc{X}^+}p_x\boldsymbol{R}(\rho_x)$.
\end{proof}
\begin{remark}
    The sets of LOCC, completely PPT-preserving, and selective LOCC channels are contained in the set of selective PPT operations, implying that every multipartite $\boldsymbol{D}$-Rains entanglement measure is an LOCC, PPT, and selective LOCC monotone.
\end{remark}

\section{Other Properties of GMRE}

\begin{proposition}
\label{prop:multi_rains_nonnegative}
    The GMRE $R(\rho)$ is nonnegative for all $\rho\in\mc{S}(\hilb_k)$, and it is equal to zero for all biseparable states.
\end{proposition}

\begin{proof}
    See Appendix~\ref{app:nonnegativity}.
\end{proof}

\subsection{Lower bound on the GMRE and sandwiched R\'enyi--Rains entanglement}

Let $\Phi^d\in\mc{S}(\hilb^{\otimes n})$ with $\dim(\hilb)=d$ denote the $n$-qudit GHZ state given by $\Phi^d\coloneq\frac{1}{d}\sum_{i,j=0}^{d-1}\op{i}{j}^{\otimes n}$.

\begin{lemma} \label{multipartite_rains_lower_bound_lemma}
    Let $\rho\in\hilb_k$ where $\dim(\hilb_{A_i})=d\ \fall i\in\{1,\ldots,k\}$. Let $F\equiv F(\Phi^d,\rho)$ denote the fidelity~\eqref{def:fidelity} of the GHZ state and $\rho$, and $h_2$ the binary entropy~\eqref{def:binary_entropy}. Suppose that $F\geq 1/d$. Then the GMRE satisfies the following inequality for all $\rho\in\mc{S}(\hilb_k)$:
    \begin{align}\label{res:multi_rains_lower_bound}
        R(\rho)& \geq D\bp{(F,1-F)\middle\|(1/d,1-1/d)} \\
        & \geq F\log_2d - h_2(F).
    \end{align}
\end{lemma}

\begin{proof}
    Define a measurement channel $\mc{M}$ with action
    \begin{align}
        \mc{M}(\cdot)&\coloneq\Tr\bb{\Phi^d(\cdot)}\op{1}{1}+\Tr\bb{(\ii^d-\Phi^d)(\cdot)}\op{0}{0}.
        \label{def:measurement_channel_m}
    \end{align}
    Because $\Phi^d$ is a pure state, $F=F(\Phi^d,\rho)=\Tr[\Phi^d\rho]$. Thus, we can write $\mc{M}(\rho)=F\op{1}{1}+(1-F)\op{0}{0}$. Let $\sigma\in\mc{T}(\hilb_k)$ and consider the following:
    \begin{align}
        D(\rho\|\sigma)&\geq D(\mc{M}(\rho)\|\mc{M}(\sigma)) \notag\\
        &=D\bp{(F,1-F)\middle\|(\Tr[\Phi^d\sigma],\Tr[\sigma]-\Tr\bb{\Phi^d\sigma})} \notag\\
        &\geq D\bp{(F,1-F)\middle\|(\Tr[\Phi^d\sigma],1-\Tr\bb{\Phi^d\sigma})} \notag\\
        &\geq D\bp{(F,1-F)\middle\|(1/d,1-1/d)} \notag\\
        &\geq -h_2(F)-F\log_2\bp{1/d},\label{ghz_relative_entropy_ineq}
    \end{align}
    where the first inequality follows from~\eqref{res:data_processing_ineq} and the equality follows from~\eqref{def:measurement_channel_m} (i.e.,  rewriting the expression on the first line in terms of the classical relative entropy). The second inequality follows because $\Tr[\sigma]\leq1$ (from~\eqref{trace_element_t_leq_1}) and from monotonicity of $\log_2$. The penultimate inequality follows because $\Tr[\Phi^d\sigma]\leq\frac{1}{d}$, as shown below, from the assumption that $F\geq 1/d$, and from a basic property of binary relative entropy (see Lemma~\ref{lem:mono-rel-ent-2nd} in Appendix~\ref{app:monotonicity-props-binary-rel-ent}). The last inequality follows by writing out the binary relative entropy and dropping a non-negative term:
    indeed $\sigma\geq0$ and $\Tr[\Phi^d\sigma]\leq\frac{1}{d}\leq 1$.

    Let us now prove that $\Tr[\Phi^d\sigma]\leq\frac{1}{d}$. Let $\sigma=\sum_mq_m\sigma_m$ be an arbitrary decomposition of $\sigma \in\mc{T}(\hilb_k)$, with $m$ labeling the bipartition. Then, 
    \begin{align}
        \Tr[\Phi^d\sigma]&=\sum_mq_m\Tr\bb{\Phi^d\sigma_m} \notag\\
        &=\sum_mq_m\Tr\bb{\Phi^dT_m(T_m(\sigma_m))} \notag\\
        &=\sum_mq_m\Tr\bb{T_m(\Phi^d)T_m(\sigma_m)} \notag\\
        &=\sum_m\frac{q_m}{d}\Tr\bb{F_mT_m(\sigma_m)} \notag\\
        &\leq\frac{1}{d}\sum_mq_m\norm{T_m(\sigma_m)}_1\leq\frac{1}{d}.
        \label{ghz_trace_ineq}
    \end{align}
    Here, we used the fact that the partial transpose is self-inverse and self-adjoint, and that $\sigma\in\mc{T}$. Also, $T_m(\Phi^d)=\frac{1}{d}F_m$ where $F_m\coloneq\sum_{k,j}\op{k,\ldots k, j, \ldots, j}{j,\ldots,j , k, \ldots, k}$ is the swap operator acting upon subsystems determined by the partition $m$ (see~\cite[Eq.~(2.5.10)]{khatri2024}). Because $\left\|F_m\right\|_\infty \leq 1$, we can apply H\"older duality to conclude the penultimate inequality (i.e., $\norm{X}_1=\sup_{Y \in\mc{L}(\hilb): \left\| Y\right\|_\infty \leq 1 }\left|\Tr[Y X]\right|$).
\end{proof}

\begin{lemma}\label{lem:multi_rains_ghz}
    Let $\Phi^d$ be a GHZ state. Then,
    \begin{align}
        R(\Phi^d)&=\log_2 d.
    \end{align}
\end{lemma}

\begin{proof}
    See Appendix~\ref{app:multi_rains_ghz}.
\end{proof}

The following upper bound for the multipartite R\'enyi--Rains entanglement generalizes the bound from~\cite[Eq.~(25)]{eisert2022}.

\begin{lemma} \label{lem:multi_sandwich_renyi_rains_lower_bound}
    Let $\rho\in\mc{S}(\hilb_k)$ where $\dim(\hilb_{A_i})=d\ \fall i\in\{1,\ldots,k\}$, and let $F\equiv F(\Phi^d,\rho)$ denote the fidelity of the GHZ state $\Phi^d$ and $\rho$. Then, for all $\alpha>1$, the multipartite R\'enyi--Rains entanglement satisfies the following for all $\rho\in\mc{S}(\hilb_k)$:
    \begin{align}
        \widetilde{R}_\alpha(\rho)-\frac{\alpha}{\alpha-1}\log_2 F&\geq \log_2d.
    \end{align}
\end{lemma}

\begin{proof}
    See Appendix~\ref{app:multi_sandwich_renyi_rains_lower_bound}.
\end{proof}

% ----====----====----==== GHZ-distilable entanglement ====----====----====----
\section{Upper bound on GHZ-distillable entanglement}

We now show that the GMRE is an upper bound on the one-shot GHZ-distillable entanglement and the GHZ probabilistic approximate distillable multipartite entanglement (GHZ-PADME). Given $\ep\in[0,1]$, the one-shot $\ep$-GHZ-distillable entanglement of $\rho\in\mc{S}(\hilb_k)$ is defined as
\begin{align}
    E^\ep_{\operatorname{D}}(\rho)&\coloneq\sup_{(d,\mc{L}^\leftrightarrow)}\bc{\log_2d:F(\Phi^d,\mc{L}(\rho))\geq1-\ep},
    \label{eq:def-1-shot-GHZ-distill-ent}
\end{align}
where $d\geq1$, $\mc{L}^{\leftrightarrow}\in\textrm{LOCC}(\hilb_k,\hilb'^{\otimes k})$, and $\textrm{dim}(\hilb')=d$. The one-shot $\ep$-GHZ-PADME of $\rho$ is defined as
\begin{multline}
    E^\ep_{\operatorname{PD}}(\rho)\coloneq\sup_{(d,\mc{L}^\leftrightarrow,p\in[0,1])}\bigl\{p\log_2d:\mc{L}^\leftrightarrow(\rho)=p\op{1}{1}\otimes\widetilde{\Phi}^d\\
    +(1-p)\op{0}{0}\otimes\sigma,F(\Phi^d,\widetilde{\Phi}^d)\geq1-\ep,\sigma\in\mc{S}(\hilb'^{\otimes k})\bigr\}.\nonumber
\end{multline}
Observe that $E^\ep_{\operatorname{D}}\leq E^\ep_{\operatorname{PD}}$ because we can set $p=1$ in the definition of $E^\ep_{\operatorname{PD}}$.

\begin{theorem}\label{thm:one_shot_pade_bound}
    Let $\rho\in\mc{S}(\hilb_k)$ and $\ep\in[0,1)$. Then, the following bounds hold for the one-shot  $\ep$-distillable entanglement and $\ep$-GHZ-PADME of $\rho$:
    \begin{align}
        E^\ep_{\operatorname{D}}\leq E^\ep_{\operatorname{PD}}(\rho)& \leq\frac{1}{1-\ep}\bp{R(\rho)+h_2(\ep)} ,\label{res:one_shot_pade_bound} \\
        E^\ep_{\operatorname{D}}\leq E^\ep_{\operatorname{PD}}(\rho)& \leq
        \inf_{\alpha>1} \left\{\widetilde{R}_\alpha(\rho)+\frac{\alpha}{\alpha-1}\log_2\!\left(\frac{1}{1-\ep}\right)\right\}.
        \label{res:one_shot_pade_bound-alpha-bnd}
    \end{align}
\end{theorem}

\begin{proof}
Here we prove \eqref{res:one_shot_pade_bound} and provide a proof of \eqref{res:one_shot_pade_bound-alpha-bnd} in Appendix~\ref{app:asymptotic_pade_bound}.
     Let $(d,\mc{L},p)$ define an arbitrary GHZ-PADME protocol with $d\geq1$, $\mc{L}\in\textrm{LOCC}(\hilb_m,\hilb_X\otimes\hilb'^{\otimes k})$, $\dim(\hilb')=d$, and $p\in[0,1]$. Let $F\equiv F(\Phi^d,\widetilde{\Phi}^d)$. Let us first suppose that $1-\varepsilon \leq \frac{1}{d}$. Then it follows that
    \begin{align}
    p \log_2 d & \leq p \log_2\!\left(\frac{1}{1-\varepsilon}\right)   
     \leq \frac{1}{1-\ep}\bp{R(\rho)+h_2(\ep)},
    \end{align}
    where the second inequality follows because $p \leq 1$, $-\log_2\!\left({1-\varepsilon}\right) \leq \frac{1}{1-\ep}h_2(\ep)$, and $R(\rho)\geq 0$, by Proposition~\ref{prop:multi_rains_nonnegative}.
    Now suppose that $1-\varepsilon \geq \frac{1}{d}$, so that $F \geq 1-\ep \geq 1/d$. Then
    \begin{align}
        R(\rho)&\geq pR(\widetilde{\Phi}^d)+(1-p)R(\sigma) \notag\\
        &\geq pR(\widetilde{\Phi}^d) \notag\\
        & \geq p D((F,1-F)\|(1/d,1-1/d))\notag \\
        & \geq p D((1-\ep,\ep)\|(1/d,1-1/d)) \notag\\
        &\geq p(1-\ep)\log_2d-h_2(\ep),\label{intres:one_shot_distill}
    \end{align}
    where the first inequality follows from Theorem~\ref{rains_ppt_monotone_thm}, the second from Proposition~\ref{prop:multi_rains_nonnegative}, the third from Lemma~\ref{multipartite_rains_lower_bound_lemma}, the fourth from a basic property of binary relative entropy (see Lemma~\ref{lem:mono-rel-ent-1st} in Appendix~\ref{app:monotonicity-props-binary-rel-ent}), and the monotonicity of $\log_2$ was used. We then conclude that $p\log_2d\leq\frac{1}{1-\ep}(R(\rho)+h_2(\ep))$ for all $\ep \in [0,1)$. Because the distillation protocol is arbitrary, it follows that $E^\ep_{\operatorname{PD}}(\rho)\leq\frac{1}{1-\ep}(R(\rho)+h_2(\ep))$, concluding the proof.
\end{proof}
Let us define the asymptotic GHZ-PADME and strong converse GHZ-PADME of $\rho\in\mc{S}(\hilb_k)$, respectively, as
\begin{align}
    E_{\operatorname{PD}}(\rho)&\coloneq\inf_{\ep\in(0,1]}\liminf_{n\ra\infty}\frac{1}{n}E^\ep_{\operatorname{PD}}(\rho^{\otimes n}),\label{def:pade}\\
    \widetilde{E}_{\operatorname{PD}}(\rho)&\coloneq\sup_{\ep\in(0,1]}\limsup_{n\ra\infty}\frac{1}{n}E^\ep_{\operatorname{PD}}(\rho^{\otimes n}).\label{def:strong_converse_pade}
\end{align}

\begin{theorem}
\label{thm:asymptotic_pade_bound}
    The regularized GMRE is a weak converse upper bound for the asymptotic GHZ-PADME. That is, for $\rho\in\mc{S}(\hilb_k)$,
    \begin{align}
    \label{res:asymptotic_distillable_entanglement_bound}
        E_{\operatorname{PD}}(\rho)\leq \liminf_{n\to\infty} \frac{1}{n} R(\rho^{\otimes n }).
    \end{align}
    Furthermore, we have the following regularized upper bound on the strong converse GHZ-PADME:
    \begin{equation}
        \widetilde{E}_{\operatorname{PD}}(\rho)\leq \inf_{\alpha > 1} \limsup_{n\ra\infty}\frac{1}{n}\widetilde{R}_\alpha(\rho^{\otimes n})
        \label{res:asymptotic_SC_distillable_entanglement_bound}.
    \end{equation}
\end{theorem}

\begin{proof}
    See Appendix~\ref{app:asymptotic_pade_bound}.
\end{proof}

% ----====----====----==== SDP ====----====----====----
\section{Semi-definite constraint for GMRE}
Following the approach of~\cite[Lemma~12]{fang_wang_tomamichel_duan2019}, we can write $\sum_m\norm{T_m(\widetilde{\tau}_m)}_1\leq1$ as a semi-definite constraint.
\begin{lemma}\label{lem:T_semidef_constraint}
    Suppose $\tau=\sum_m\widetilde{\tau}_m$ with $\widetilde{\tau}_m\geq0$ for all $m$. Then, $\tau\in\mc{T}$ if and only if, for each $m$, there exist $\widetilde{\tau}_m^+,\widetilde{\tau}_m^-\geq0$ such that $T_m(\widetilde{\tau}_m)=\widetilde{\tau}_m^+-\widetilde{\tau}_m^-$ and $\sum_m\Tr[\widetilde{\tau}_m^++\widetilde{\tau}_m^-]\leq1$.
\end{lemma}
\begin{proof}
    See Appendix~\ref{app:T_semidef_constraint}.
\end{proof}
While the quantum relative entropy is nonlinear, there exist various numerical minimization techniques~\cite{koßmann_schwonnek2025,he_saunderson_fawzi2025,cvxquad} for this quantity that can be used with Lemma~\ref{lem:T_semidef_constraint} to write an SDP for the GMRE. In the ancillary files of our paper's arXiv posting, we provide MATLAB code for calculating the GMRE that uses~\cite{cvx,cvxquad,qetlab} (see also Appendix~\ref{app:multi_rains_code}).

% ----====----====----==== conclusion ====----====----====----
\section{Conclusion}

In this paper, we introduced the GMRE as a measure of genuine multipartite entanglement that satisfies selective PPT monotonicity, provides an upper bound on the one-shot distillable entanglement and PADME of GHZ states, and can be computed using semi-definite programming. Future research directions include investigating how to approximate or calculate the GMRE on a quantum computer or seeing what simplifications arise in calculating it by making assumptions, such as symmetry, about the underlying states.

% ----====----====----==== references ====----====----====----
%\newpage
\bibliographystyle{IEEEtran}
\bibliography{references}

% ----====----====----==== appendices ====----====----====----
\newpage
\appendices
\numberwithin{equation}{section}

\section{Proof of Theorem~\ref{thm:general_convex_opt_theorem}}\label{app:multi_rains_convex}

\begin{proof}
Our goal is to prove that $\boldsymbol{R}(\rho)=\inf_{\tau\in\mc{T}(\hilb_{k})}\boldsymbol{D}(\rho\|\tau)$, and our approach will be similar to that of~\cite{audenaert2002}. We will first show $\boldsymbol{R}(\rho)\geq\inf_{\tau\in\mc{T}(\hilb_{k})}\boldsymbol{D}(\rho\|\tau)$ and then $\boldsymbol{R}(\rho)\leq\inf_{\tau\in\mc{T}(\hilb_{k})}\boldsymbol{D}(\rho\|\tau)$, where the set $\mathcal{T}$ is defined in~\eqref{def:general_T_set_definition}.

Let $\rho\in \mc{S}(\hilb_{k})$, and let
\begin{align}
    \mc{S}'&\coloneq\{\sigma\in\mc{L}_+(\hilb_{k}),\Tr[\sigma]\leq1\}.
\end{align}
Consider the following:
\begin{align}
    \boldsymbol{R}(\rho)&=\inf_{\substack{\sigma\in\mc{S}(\hilb_{k}) \\ \sigma=\sum_mr_m\omega_m}}\!\!\bc{\boldsymbol{D}\bp{\rho\middle\|\sigma}+\log_2\bp{\sum_mr_m\norm{T_m(\omega_m)}_1}\!} \notag\\
    &=\inf_{\substack{\sigma\in\mc{S}'(\hilb_{k}) \\ \sigma=\sum_mr_m\omega_m}}\Biggl\{\boldsymbol{D}\bp{\rho\middle\|\sigma}+\log_2\bp{\sum_mr_m\norm{T_m(\omega_m)}_1}\notag\\
    &\hspace{3cm}-\log_2(\Tr[\sigma])\Biggr\}.
\end{align}
We can extend the set of states we are minimizing over to $\mc{S}'$ because the first two terms do not depend on $\Tr[\sigma]$ (indeed they are invariant under the rescaling $\sigma \to c \sigma$ for $c>0$)  and the additional term on the second line achieves a minimal value for $\sigma \in \mc{S}'(\hilb_{k})$ when $\Tr[\sigma]=1$. Using the scaling property of $\boldsymbol{D}$, we find that
\begin{align}
    \boldsymbol{R}(\rho)&=\!\!\inf_{\substack{\sigma\in\mc{S}'(\hilb_{k}) \\ \sigma=\sum_mr_m\omega_m}}\!\!\bc{\boldsymbol{D}\bp{\rho\middle\|\sigma}-\log_2\bp{\frac{\Tr[\sigma]}{\sum_mr_m\norm{T_m(\omega_m)}_1}}\!} \notag\\
    &=\inf_{\substack{\sigma\in\mc{S}'(\hilb_{k}) \\ \sigma=\sum_mr_m\omega_m}}\bc{\boldsymbol{D}\left(\rho\middle\|\frac{\Tr[\sigma]}{\sum_mr_m\norm{T_m(\omega_m)}_1}\sigma\right)} \notag\\
    &\geq\inf_{\tau\in\mc{T}(\hilb_{k})}\bc{\boldsymbol{D}(\rho\|\tau)},\label{rains_inequality_1}
\end{align}
where the inequality follows because $\frac{\Tr[\sigma]}{\sum_mr_m\norm{T_m(\omega_m)}_1}\sigma\in\mc{T}$. To see this, consider that $\sigma\geq0$ and
\begin{align}
    \frac{\Tr[\sigma]}{\sum_mr_m\norm{T_m(\omega_m)}_1}&\geq0.
\end{align}
We also have
\begin{align}
    \sum_\ell&\frac{\Tr[\sigma]}{\sum_mr_m\norm{T_m(\omega_m)}_1}r_\ell\norm{T_\ell(\omega_\ell)}_1 \notag\\
    &=\frac{\Tr[\sigma]}{\sum_mr_m\norm{T_m(\omega_m)}_1}\sum_\ell r_\ell\norm{T_\ell(\omega_\ell)}_1 \notag\\
    &=\Tr[\sigma]\leq1.
\end{align}
Thus, $\frac{\Tr[\sigma]}{\sum_mr_m\norm{T_m(\omega_m)}_1}\sigma\in\mc{T}$, and we conclude the first desired inequality:
\begin{align}
    \boldsymbol{R}(\rho)\geq\inf_{\tau\in\mc{T}(\hilb_{k})}\boldsymbol{D}(\rho\|\tau).\label{intres:rains_convex_opt_ineq1}
\end{align}

We will now prove that $\boldsymbol{R}(\rho)\leq\inf_{\tau\in\mc{T}(\hilb_{k})}\boldsymbol{D}(\rho\|\tau)$. Let $\tau\in\mc{T}$ where $\tau=\sum_mq_m\tau_m$. Then,
\begin{align}
    \boldsymbol{D}(\rho\|\tau)&=\boldsymbol{D}(\rho\|\tau)-\log_2(\Tr[\tau])+\log_2(\Tr[\tau]) \notag\\
    &\hspace{.4cm}+\log_2\bp{\frac{1}{\Tr[\tau]}\sum_mq_m\norm{T_m(\tau_m)}_1} \notag\\
    &\hspace{.4cm}-\log_2\bp{\frac{1}{\Tr[\tau]}\sum_mq_m\norm{T_m(\tau_m)}_1} \notag\\
    &=\boldsymbol{D}\bp{\rho\middle\|\frac{\tau}{\Tr[\tau]}}+\log_2\bp{\frac{1}{\Tr[\tau]}\sum_mq_m\norm{T_m(\tau_m)}_1} \notag\\
    &\hspace{.4cm}-\log_2\bp{\sum_mq_m\norm{T_m(\tau_m)}_1} \notag \\
    & \geq
    \boldsymbol{D}\bp{\rho\middle\|\frac{\tau}{\Tr[\tau]}}+\log_2\bp{\frac{1}{\Tr[\tau]}\sum_mq_m\norm{T_m(\tau_m)}_1\!}\!,
    \label{rains_inequality_2_pt1}
\end{align}
where~\eqref{rains_inequality_2_pt1} follows because $\tau\in\mc{T}$ by assumption. We thus conclude that
\begin{align}
     \boldsymbol{D}(\rho\|\tau)&\geq \boldsymbol{D}\bp{\rho\middle\|\frac{\tau}{\Tr[\tau]}}+\log_2\bp{\frac{1}{\Tr[\tau]}\sum_mq_m\norm{T_m(\tau_m)}_1\!}\!.
     \label{rains_inequality_2_pt1-1}
\end{align}
Let $\tau_m'\coloneqq \frac{\tau_m}{\Tr[\tau_m]}$ and $q_m'\coloneqq q_m\Tr[\tau_m]$. Then, $\sum_mq_m'\tau_m'=\sum_mq_m\tau_m$ and $\sum_mq_m'=\Tr[\tau]$. We can then write
\begin{align}
    &\boldsymbol{D}\bp{\rho\middle\|\frac{\tau}{\Tr[\tau]}}+\log_2\bp{\frac{1}{\Tr[\tau]}\sum_mq_m\norm{T_m(\tau_m)}_1} \notag\\
    &=\boldsymbol{D}\bp{\rho\middle\|\frac{\tau}{\Tr[\tau]}}+\log_2\bp{\frac{1}{\Tr[\tau]}\sum_mq_m'\norm{T_m(\tau_m')}_1} \notag\\
    &\geq \boldsymbol{R}(\rho).\label{rains_inequality_2_pt2}
\end{align}
Note that $\bp{\frac{q_m'}{\Tr[\tau]}}_m$ is a probability distribution because $q_m'\geq0$, $\Tr[\tau]\geq0$, and $\frac{1}{\Tr[\tau]}\sum_mq_m'=1$, using the definition of $q_m'$. Because $\tau$ is arbitrary, we can take an infimum over $\tau$ in~\eqref{rains_inequality_2_pt1-1} and use~\eqref{rains_inequality_2_pt2} to conclude that $ \boldsymbol{R}(\rho)\leq\inf_{\tau\in\mc{T}(\hilb_{k})}\boldsymbol{D}(\rho\|\tau)$. With this and~\eqref{intres:rains_convex_opt_ineq1}, we conclude $\boldsymbol{R}(\rho)=\inf_{\tau\in\mc{T}(\hilb_{k})}\boldsymbol{D}(\rho\|\tau)$.
\end{proof}

\section{Comparison of GMRE and genuine multipartite log-negativity in the transverse field Ising model}
\label{app:multi_rains_and_gmn}

\begin{figure}[h]
\centerline{\includegraphics[width=9cm]{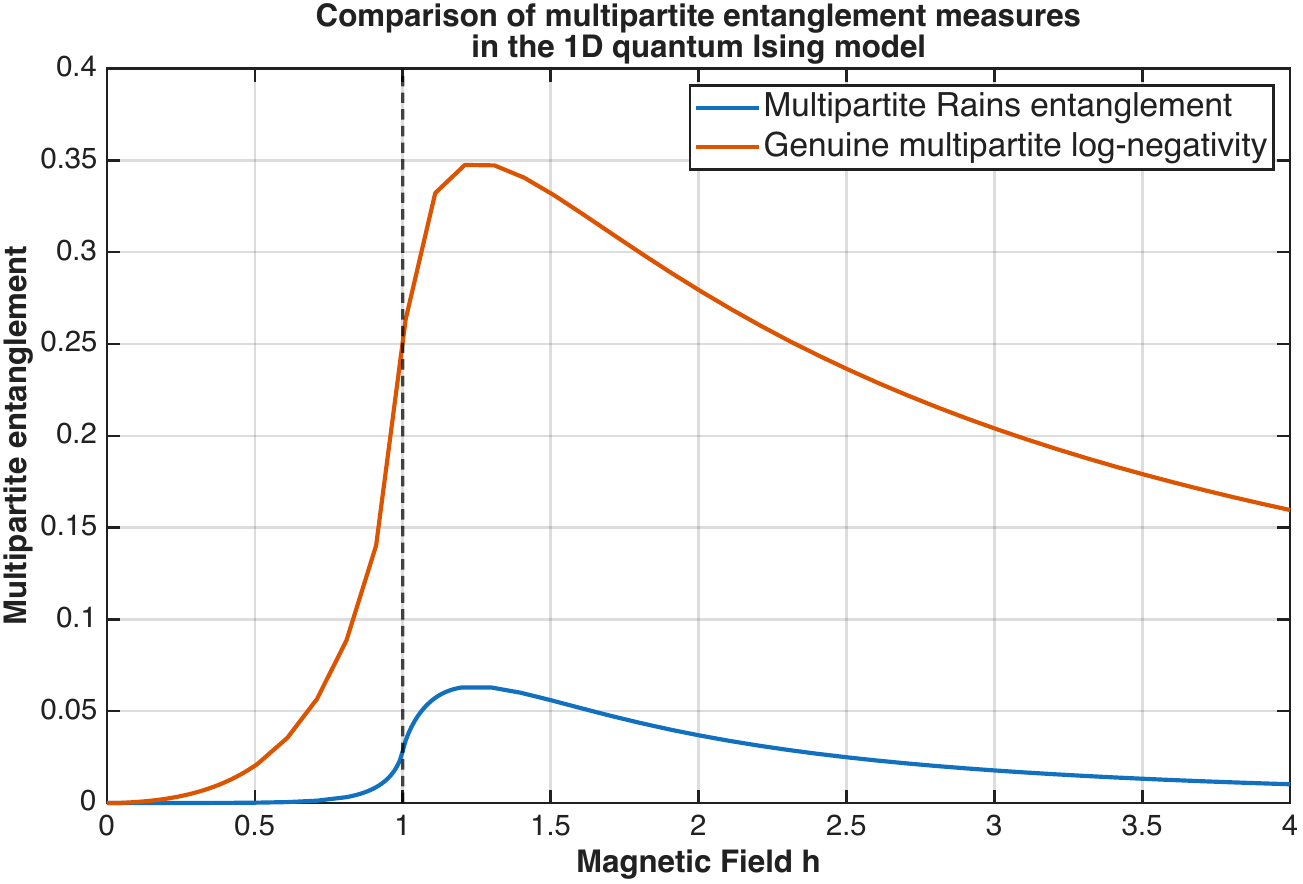}}
\caption{The GMRE and genuine multipartite log-negativity for tripartite reduced density matrices in the 1D transverse field Ising model are plotted versus the magnetic field $h$. 
The reduced density matrices are obtained exactly via the Jordan-Wigner transformation~\cite{Fagotti2013RDM},
and the code for calculating the entanglement uses~\cite{cvx,cvxquad,qetlab}. We used the relationship $E_N=\log_2(2N+1)$ to calculate the log-negativity $E_N$, where~\cite[Eq.~(8)]{hofmann2014} was used for calculating the genuine multipartite negativity $N$.}
\label{fig:multipartite_measures_vs_h}
\end{figure}

This appendix presents an application of the genuine multipartite Rains entanglement in a standard condensed-matter setting and uses it to contrast GMRE with the genuine multipartite log-negativity (log-GMN) for the same family of states. We consider the one-dimensional transverse field Ising model (TFIM),
\begin{equation}
H \;=\; -\sum_{i} (\sigma_i^{x}\sigma_{i+1}^{x} + h\sigma_i^{z}),
\label{eq:TFIM_Hamiltonian}
\end{equation}
where $\sigma_i^\alpha$ is a Pauli operator acting on site $i$, and $h$ is the transverse field strength (the Ising coupling is set to unity). The model undergoes a quantum phase transition at $h=1$, where quantum fluctuations are strongest and correlations become scale invariant.
We focus on tripartite reduced density matrices extracted from the TFIM ground state, obtained exactly via the Jordan--Wigner transformation~\cite{Fagotti2013RDM}.

Figure~\ref{fig:multipartite_measures_vs_h} compares the GMRE and log-GMN as a function of $h$ for the tripartite reduced density matrix of three adjacent spins. Both measures are enhanced in the vicinity of the critical point and peak near $h\approx 1$, consistent with previous literature~\cite{wang_etal2025,Lyu2025Ising}.
At the same time, the two quantities are clearly separated across the full field range: The GMRE is uniformly smaller and exhibits a markedly different field dependence. In particular, compared with log-GMN, the multipartite Rains entanglement decays much faster at small field strength, becoming vanishingly small already around $h\approx 0.5$, while log-GMN remains much larger.

This suggests a practical utility of the GMRE in condensed-matter applications: While negativity-based diagnostics and related GME measures studied previously in the one-dimensional Ising model show qualitatively similar trends versus $h$, the GMRE appears substantially more selective in the low-field regime. In this sense, the GMRE has the potential to be a sharper indicator of quantum criticality.

\section{GMRE as the limit of the genuine multipartite sandwiched R\'enyi--Rains entanglement}

\label{app:sandwiched-RR-to-GRME}

\begin{proposition}
\label{prop:multi_rains_limit_sandwiched_renyi_rains}
    The GMRE is the limit of the multipartite R\'enyi--Rains entanglement for $\alpha\ra1$. That is, for all $\rho\in\mc{S}(\hilb_k)$, the multipartite sandwiched R\'enyi--Rains entanglement satisfies the following equality:
    \begin{align}\label{res:multi_rains_limit_sandwiched_renyi_rains}
        \lim_{\alpha\ra1}\widetilde{R}_\alpha(\rho)&=R(\rho).
    \end{align}
\end{proposition}

\begin{proof}
    We will first show that this limit holds from the right and then from the left. Consider the following:
    \begin{align}\label{res:limit_multi_renyi_rains_right}
    R(\rho)&=\inf_{\tau\in\mc{T}(\hilb_k)}D(\rho\|\tau) \notag\\
    &=\inf_{\tau\in\mc{T}(\hilb_k)}\lim_{\alpha\ra1^+}\widetilde{D}_\alpha(\rho\|\tau) \notag\\
    &=\inf_{\tau\in\mc{T}(\hilb_k)}\inf_{\alpha\in(1,\infty)}\widetilde{D}_\alpha(\rho\|\tau) \notag\\
    &=\inf_{\alpha\in(1,\infty)}\inf_{\tau\in\mc{T}(\hilb_k)}\widetilde{D}_\alpha(\rho\|\tau) \notag\\
    &=\lim_{\alpha\ra1^+}\widetilde{R}_\alpha(\rho),
\end{align}
where~\eqref{res:limit_sandwich_renyi_quantum_rel_entropy} was used on the second line, and we used the fact that $\widetilde{D}_\alpha$ is monotonically increasing in $\alpha$ for $\alpha\in(0,1)\cup(1,\infty)$~\cite[Prop.~7.31]{khatri2024} to replace $\lim_{\alpha\ra1^+}$ with $\inf_{\alpha\in(1,\infty)}$ on the third line. Now, we consider the other limit:
\begin{align}\label{res:limit_multi_renyi_rains_left}
    R(\rho)&=\inf_{\tau\in\mc{T}(\hilb_k)}D(\rho\|\tau) \notag\\
    &=\inf_{\tau\in\mc{T}(\hilb_k)}\lim_{\alpha\ra1^-}\widetilde{D}_\alpha(\rho\|\tau) \notag\\
    &=\inf_{\tau\in\mc{T}(\hilb_k)}\sup_{\alpha\in[1/2,1)}\widetilde{D}_\alpha(\rho\|\tau) \notag\\
     &=\sup_{\alpha\in[1/2,1)}\inf_{\tau\in\mc{T}(\hilb_k)}\widetilde{D}_\alpha(\rho\|\tau) \notag\\
      &=\lim_{\alpha\ra1^-}\widetilde{R}_\alpha(\rho).
\end{align}
Here, we used~\eqref{res:limit_sandwich_renyi_quantum_rel_entropy} on the second line. We also employed the Mosonyi--Hiai minimax theorem (presented in~\cite[Corollary A.2]{mosonyi_hiai2011}; see also~\cite[Theorem 2.26]{khatri2024}), which, along with the fact that $\mc{T}(\hilb_k)$ is compact, is applicable due to the function $(\tau,\alpha)\ra\widetilde{D}_\alpha(\rho\|\tau)$ being lower semi-continuous in $\tau$ and monotonically increasing in $\alpha$ for $\alpha\in(0,1)\cup(1,\infty)$. So, with~\eqref{res:limit_multi_renyi_rains_right} and~\eqref{res:limit_multi_renyi_rains_left}, we conclude $\lim_{\alpha\ra1}\widetilde{R}_\alpha(\rho)=R(\rho)$.
\end{proof}

\section{Proof of Proposition~\ref{prop:multi_rains_nonnegative}}
\label{app:nonnegativity}

\begin{proof}
    We will show that $R$ is nonnegative by first showing $\sigma\in\mc{T}\Rightarrow\Tr[\sigma]\leq1$ and then applying Klein's inequality, which states that if $\Tr[\sigma]\leq1$, then $D(\rho\|\sigma)\geq0$ for $\rho\in\mc{S}(\hilb)$. Let $\sigma=\sum_mq_m\sigma_m$ be an arbitrary decomposition of $\sigma$. Then,
    \begin{align}
        \Tr[\sigma]&=\sum_mq_m\Tr[\sigma_m] \notag\\
        &=\sum_mq_m\Tr[T_m(\sigma_m)] \notag\\
        &\leq\sum_mq_m\norm{T_m(\sigma_m)}_1\leq1\label{trace_element_t_leq_1}
    \end{align}
where the last inequality follows because $\sigma\in\mc{T}$. With Klein's inequality, we conclude that $D(\rho\|\sigma)\geq0$ for all $\rho\in\mc{S}(\hilb)$ and $\sigma\in\mc{T}$, implying $R(\rho)\geq0$, which completes the proof of nonnegativity.

To see that the GMRE is equal to zero for all biseparable states, note that all biseparable states are PPT mixtures with unit trace, implying that they are elements of $\mc{T}$. Thus, if $\rho$ is biseparable, $\tau=\rho$ attains $R(\rho)=0$, which is the minimum value of the GMRE since it is nonnegative.
\end{proof}

\section{Monotonicity of binary relative entropy}
\label{app:monotonicity-props-binary-rel-ent}

\begin{lemma}[Monotonicity of binary relative entropy in the first argument]\label{lem:mono-rel-ent-1st}
Let $r_0,r_1>0$ and define, for $p\in(0,1)$,
\begin{align}
f(p)
& \coloneqq  D\!\left((p,1-p)\,\middle\|\,(r_0,r_1)\right)\notag \\
& = p\log_2\bp{\frac{p}{r_0}} + (1-p)\log_2\bp{\frac{1-p}{r_1}}.
\end{align}
Set
\begin{equation}
p^\ast \coloneqq  \frac{r_0}{r_0+r_1}.    
\end{equation}
Then $f$ is strictly decreasing on $(0,p^\ast)$ and strictly increasing on $(p^\ast,1)$. In particular, $f$ attains its unique minimum at $p=p^\ast$.
\end{lemma}

\begin{proof}
The derivative of $f$ is
\begin{equation}
f'(p)
= \log_2\bp{\frac{p}{r_0}} - \log_2\bp{\frac{1-p}{r_1}}
= \log_2\!\left(\frac{p r_1}{(1-p) r_0}\right).
\end{equation}
Hence $f'(p)=0$ if and only if
\begin{equation}
\frac{p}{1-p} = \frac{r_0}{r_1},
\end{equation}
which yields the unique critical point $p^\ast = \frac{r_0}{r_0+r_1}$.

For $p<p^\ast$ we have $\frac{p}{1-p}<\frac{r_0}{r_1}$ and therefore $f'(p)<0$, while for
$p>p^\ast$ we have $f'(p)>0$. Consequently, $f$ is strictly decreasing on $(0,p^\ast)$ and strictly increasing on $(p^\ast,1)$.

Moreover,
\begin{equation}
f''(p)=\frac{1}{\ln2}\bp{\frac{1}{p}+\frac{1}{1-p}}>0 \quad \text{for all } p\in(0,1),
\end{equation}
so $f$ is strictly convex, and the critical point $p^\ast$ is the unique global minimum.
\end{proof}

\begin{lemma}[Monotonicity of binary relative entropy in the second argument]
\label{lem:mono-rel-ent-2nd}
Let $p\in(0,1)$, $s>0$, and $r\in(0,s)$. Define
\begin{align}
F(r)
& \coloneqq  D\bigl((p,1-p)\,\|\,(r,s-r)\bigr) \notag \\
& = p\log_2\bp{\frac{p}{r}} + (1-p)\log_2\bp{\frac{1-p}{s-r}}
 .
\end{align}
Then $F$ is strictly decreasing on $(0,ps)$, strictly increasing on $(ps,s)$ and attains a unique global minimum at $r=ps$.
\end{lemma}

\begin{proof}
We differentiate $F$ with respect to $r$. Writing
\begin{equation}
F(r)
= p(\log_2 p - \log_2 r) + (1-p)\bigl(\log_2(1-p) - \log_2(s-r)\bigr),
\end{equation}
we obtain
\begin{equation}
\frac{d}{dr}F(r)
= \frac{1}{\ln2}\bp{-\frac{p}{r} + \frac{1-p}{s-r}}.
\end{equation}
Combining the terms over a common denominator yields
\begin{align}
\frac{d}{dr}F(r)
& = \frac{1}{\ln2}\bp{\frac{-p(s-r) + (1-p)r}{r(s-r)}} \notag \\
& = \frac{1}{\ln2}\frac{r-ps}{r(s-r)}.
\end{align}
Since $r>0$ and $s-r>0$ for all $r\in(0,s)$, and $\ln2>0$, the denominator is strictly positive. Hence the sign of $\frac{d}{dr}F(r)$ is determined by the numerator $r-ps$. It follows that $\frac{d}{dr}F(r)<0$ for $r<ps$, $\frac{d}{dr}F(r)=0$ at $r=ps$, and $\frac{d}{dr}F(r)>0$ for $r>ps$.

Therefore, $F$ is strictly decreasing on $(0,ps)$, strictly increasing on $(ps,s)$, and has a unique global minimum at $r=ps$.
\end{proof}

We now provide more detail for how Lemmas~\ref{lem:mono-rel-ent-1st} and \ref{lem:mono-rel-ent-2nd} are used in the proofs of Lemma~\ref{multipartite_rains_lower_bound_lemma} and Theorem~\ref{thm:one_shot_pade_bound}. Let us begin by showing how Lemma~\ref{lem:mono-rel-ent-2nd} justifies the third inequality in \eqref{ghz_relative_entropy_ineq}. If we let $p\ra F$, $r\ra\Tr[\Phi^d\sigma]$, and $s\ra1$, then it follows from Lemma~\ref{lem:mono-rel-ent-2nd} that the function
\begin{align}
    f(\Tr[\Phi^d\sigma])\coloneq D\bp{(F,1-F)\middle\|\,\bp{\Tr[\Phi^d\sigma],1-\Tr\bb{\Phi^d\sigma}}}
\end{align}
is strictly decreasing on $(0,F)$. Now, we know $\Tr[\Phi^d\sigma]\leq1/d\leq F$, using \eqref{ghz_trace_ineq} along with the assumption made in Lemma~\ref{multipartite_rains_lower_bound_lemma}, that $F\geq1/d$. With this, we arrive at
\begin{multline}
    D\bp{(F,1-F)\middle\|(\Tr[\Phi^d\sigma],1-\Tr\bb{\Phi^d\sigma})}\\
    \geq D\bp{(F,1-F)\middle\|(1/d,1-1/d)}.
\end{multline}

We now turn our focus to  Theorem~\ref{thm:one_shot_pade_bound}, where we show how we arrive at the fourth inequality in \eqref{intres:one_shot_distill} using Lemma~\ref{lem:mono-rel-ent-1st}. Let $p\ra F$, $r_0\ra1/d$, and $r_1\ra1-1/d$. It then follows from Lemma~\ref{lem:mono-rel-ent-1st} that the function
\begin{align}
    f(F)\coloneq D\bp{(F,1-F)\middle\|(1/d,1-1/d)}
\end{align}
is strictly increasing on $(1/d,1)$. Under the assumption before \eqref{intres:one_shot_distill} that $1-\ep\geq1/d$, we have that $F\geq1-\ep\geq1/d$. Thus,
\begin{multline}
    D\bp{(F,1-F)\middle\|(1/d,1-1/d)}\\
    \geq D\bp{(1-\ep,\ep)\middle\|(1/d,1-1/d)},
\end{multline}
which justifies the fourth inequality in \eqref{intres:one_shot_distill}.

\section{Proof of Lemma~\ref{lem:multi_rains_ghz}}
\label{app:multi_rains_ghz}

\begin{proof}
    We will prove this statement by first showing $R(\Phi^d)\geq\log_2 d$ and then showing $R(\Phi^d)\leq\log_2 d$.

    It is straightforward to prove the first inequality using Lemma~\ref{multipartite_rains_lower_bound_lemma}:
    \begin{align}
        R(\Phi^d)&\geq F(\Phi^d,\Phi^d)\log_2d-h_2(F(\Phi^d,\Phi^d)) \notag\\
        &=\log_2d-h_2(1)=\log_2d\label{ghz_entropy_ineq_1},
    \end{align}
    where we used the fact that $\lim_{x\ra0^+}x\log_2x=0$.

    Now, for the second inequality, let us first consider the completely dephased GHZ state $\overline{\Phi}^d\in\mc{S}(\hilb^{\otimes n})$, with $\dim(\hilb)=d$, given by
    \begin{align}
        \overline{\Phi}^d&=\frac{1}{d}\sum_{i=0}^{d-1}\op{i}{i}^{\otimes n}
    \end{align}
    where $\{\ket{i}\}_{i=0}^{d-1}$ is the computational basis of $\mc{H}$. We will now show that $\overline{\Phi}^d\in\mc{T}$. Because $\op{i}{i}^{\otimes n}\geq0$ for all $i\in\{0,1,...,d-1\}$, we need only show that $\overline{\Phi}^d$ satisfies the final condition stipulated in \eqref{def:general_T_set_definition}. Consider the following:
    \begin{align}
        \frac{1}{d}\sum_{i=0}^{d-1}\norm{T_m\bp{\op{i}{i}^{\otimes n}}}_1&=\frac{1}{d}\sum_{i=0}^{d-1}\norm{\op{i}{i}^{\otimes n}}_1 \notag\\
        &=\frac{1}{d}\sum_{i=0}^{d-1}\Tr\bb{\op{i}{i}^{\otimes n}} \notag\\
        &=\Tr\bb{\overline{\Phi}^d} \notag\\
        &=1.
    \end{align}
    Thus, $\overline{\Phi}^d\in\mc{T}$. Now, let us consider $D(\Phi^d\|\overline{\Phi}^d)$:
    \begin{align}
        D(\Phi^d\|\overline{\Phi}^d)&=-H(\Phi^d)-\Tr\bb{\Phi^d\log_2\overline{\Phi}^d} \notag\\
        &=-\Tr\bb{\Phi^d\sum_{i=0}^{d-1}\log_2\bp{\frac{1}{d}}\op{i}{i}^{\otimes n}} \notag\\
        &=\log_2(d)\Tr\bb{\Phi^d\sum_{i=0}^{d-1}\op{i}{i}^{\otimes n}} \notag\\
        &=\log_2(d)\Tr\bb{\frac{1}{d}\sum_{k,j=0}^{d-1}\op{k}{j}^{\otimes n}\sum_{i=0}^{d-1}\op{i}{i}^{\otimes n}} \notag\\
        &=\log_2(d)\frac{1}{d}\sum_{i=0}^{d-1}\bra{i}^{\otimes n}\ket{i}^{\otimes n} \notag\\
        &=\log_2d,
    \end{align}
    where $H(\rho)\coloneq-\Tr[\rho\log_2\rho]$ is the quantum entropy of $\rho$, and we used the fact that the quantum entropy of any pure state is zero. Taking an infimum over elements of $\mc{T}$, we see that $R(\Phi^d)\leq\log_2d$, which we can use with~\eqref{ghz_entropy_ineq_1} to conclude $R(\Phi^d)=\log_2d$, completing the proof.
\end{proof}

\section{Proof of Lemma~\ref{lem:multi_sandwich_renyi_rains_lower_bound}}
\label{app:multi_sandwich_renyi_rains_lower_bound}

\begin{proof}
    Let $\mc{M}$ be a measurement channel as defined in~\eqref{def:measurement_channel_m}, and let $\sigma\in\mc{T}(\hilb_k)$. Then,
    \begin{align}
        \widetilde{D}_\alpha(\rho\|\sigma)&\geq\widetilde{D}_\alpha(\mc{M}(\rho)\|\mc{M}(\sigma)) \notag\\
        &=\widetilde{D}_\alpha\bp{(F,1-F)\middle\|\,\bp{\Tr\bb{\Phi^d\sigma},\Tr[\sigma]-\Tr\bb{\Phi^d\sigma}}} \notag\\
        &=\frac{1}{\alpha-1}\log_2\biggl(F^\alpha\bp{\Tr[\Phi^d\sigma]}^{1-\alpha} \notag\\
        &\qquad +(1-F)^\alpha\bp{\Tr\bb{\sigma}-\Tr\bb{\Phi^d\sigma}}^{1-\alpha}\biggr) \notag\\
        &\geq\frac{1}{\alpha-1}\log_2\bp{F^\alpha\bp{\Tr[\Phi^d\sigma]}^{1-\alpha}} \notag\\
        &\geq\frac{1}{\alpha-1}\log_2\bp{F^\alpha d^{\alpha-1}} \notag\\
        &=\frac{\alpha}{\alpha-1}\log_2 F+\log_2 d
    \end{align}
    where we used~\eqref{res:data_processing_sandwich_renyi}, the monotonicity of $\log_2$, and~\eqref{ghz_trace_ineq}, which states that $\Tr\bb{\Phi^d\sigma}\leq\frac{1}{d}$. Since $\sigma$ is an arbitrary element of $\mc{T}$, we conclude that
    \begin{align}
        \widetilde{R}_\alpha(\rho)-\frac{\alpha}{\alpha-1}\log_2(F)&\geq \log_2d,
    \end{align}
    completing the proof.
\end{proof}

\section{Proof of Theorem~\ref{thm:asymptotic_pade_bound}}
\label{app:asymptotic_pade_bound}

Before we prove Theorem~\ref{thm:asymptotic_pade_bound}, we will prove that the multipartite R\'enyi--Rains entanglement also provides an upper bound on the GHZ-PADE. Let us restate \eqref{res:one_shot_pade_bound-alpha-bnd} as the following lemma:

\begin{lemma}\label{lem:renyi_rains_pade_bound}
    Let $\rho\in\mc{S}(\hilb_k)$, $\ep\in[0,1)$, and $\alpha>1$. Then, the following bound holds for the one-shot $\ep$-distillable entanglement of $\rho$:
    \begin{align}\label{res:one_shot_ghz_distillable_entanglement_sandwich_bound}
        E^\ep_{\operatorname{PD}}(\rho)\leq\widetilde{R}_\alpha(\rho)+\frac{\alpha}{\alpha-1}\log_2\!\left(\frac{1}{1-\ep}\right).
    \end{align}
\end{lemma}

\begin{proof}
    Let $(d,\mc{L},p)$ define an arbitrary probabilistic approximate distillation protocol with $d\geq1$, $\mc{L}\in\textrm{LOCC}(\hilb_m,\hilb_X\otimes\hilb'^{\otimes k})$, $\dim(\hilb')=d$, and $p\in[0,1]$. For $\alpha>1$,
    \begin{align}
        \widetilde{R}_\alpha(\rho)&\geq p\widetilde{R}_\alpha(\widetilde{\Phi}^d)+(1-p)\widetilde{R}_\alpha(\sigma) \notag\\
        &\geq p\widetilde{R}_\alpha(\widetilde{\Phi}^d) \notag\\
        &\geq p\bp{\log_2d+\frac{\alpha}{\alpha-1}\log_2 F} \notag\\
        &\geq p\bp{\log_2d+\frac{\alpha}{\alpha-1}\log_2(1-\ep)} \notag\\
        &\geq p\log_2d+\frac{\alpha}{\alpha-1}\log_2(1-\ep).\label{intres:sandwich_renyi_rains_pade_bound}
    \end{align}
    We used the selective PPT monotonicity of the multipartite sandwiched R\'enyi--Rains entanglement for the first inequality, which is a consequence of Theorem~\ref{rains_ppt_monotone_thm}. For the second inequality, we used~\eqref{trace_element_t_leq_1} with the fact that the sandwiched R\'enyi relative entropy is nonnegative if the first argument has unit trace and the second has trace at most one (see \cite[Proposition~7.35]{khatri2024}). Also, we used Lemma~\ref{lem:multi_sandwich_renyi_rains_lower_bound} for the third inequality. Rearranging~\eqref{intres:sandwich_renyi_rains_pade_bound} gives
    \begin{align}
        p\log_2d&\leq\widetilde{R}_\alpha(\rho)-\frac{\alpha}{\alpha-1}\log_2(1-\ep).
    \end{align}
    Because the distillation protocol is arbitrary, it follows that
    \begin{align}
        E^\ep_{\operatorname{PD}}(\rho)&\leq\widetilde{R}_\alpha(\rho)-\frac{\alpha}{\alpha-1}\log_2(1-\ep),
    \end{align}
    concluding the proof.
\end{proof}

We will now  prove Theorem~\ref{thm:asymptotic_pade_bound}.

\begin{proof}[Proof of Theorem~\ref{thm:asymptotic_pade_bound}]
Using \eqref{res:one_shot_pade_bound}, we find that
\begin{align}
    \inf_{\ep\in(0,1]}&\liminf_{n\ra\infty}\frac{1}{n}E^\ep_{\operatorname{PD}}(\rho^{\otimes n}) \notag\\
        &\leq\inf_{\ep\in(0,1]}\liminf_{n\ra\infty}\frac{1}{n(1-\ep)}\bp{{R}(\rho^{\otimes n})+h_2(\ep)} \notag\\
        & =\inf_{\ep\in(0,1]}\liminf_{n\ra\infty}\frac{1}{n(1-\ep)}{R}(\rho^{\otimes n})\notag\\
        & =\inf_{\ep\in(0,1]} \frac{1}{1-\ep}\liminf_{n\ra\infty}\frac{1}{n}{R}(\rho^{\otimes n}) \notag\\
        & =\liminf_{n\ra\infty}\frac{1}{n}{R}(\rho^{\otimes n})
        .
    \end{align}
thus concluding the proof of \eqref{res:asymptotic_distillable_entanglement_bound}.

    Using Lemma~\ref{lem:renyi_rains_pade_bound}, we find that, for all $\alpha>1$,
    \begin{align}
        \sup_{\ep\in(0,1]}&\limsup_{n\ra\infty}\frac{1}{n}E^\ep_{\operatorname{PD}}(\rho^{\otimes n}) \notag\\
        &\leq\sup_{\ep\in(0,1]}\limsup_{n\ra\infty}\frac{1}{n}\bp{\widetilde{R}_\alpha(\rho^{\otimes n})-\frac{\alpha}{\alpha-1}\log_2(1-\ep)} \notag\\
        &=\sup_{\ep\in(0,1]}\limsup_{n\ra\infty}\frac{1}{n}\widetilde{R}_\alpha(\rho^{\otimes n}) \notag\\
        &=\limsup_{n\ra\infty}\frac{1}{n}\widetilde{R}_\alpha(\rho^{\otimes n}) .
    \end{align}
    Because this holds for all $\alpha> 1$, we thus conclude \eqref{res:asymptotic_SC_distillable_entanglement_bound} after taking the infimum over $\alpha > 1$.
\end{proof}

\section{Proof of Lemma~\ref{lem:T_semidef_constraint}}
\label{app:T_semidef_constraint}

\begin{proof}
    For $\tau\in\mc{L}_+(\hilb_k)$ with $\tau=\sum_m\widetilde{\tau}_m$ and each $\widetilde{\tau}_m\geq0$, our goal is to show $\tau\in\mc{T}(\hilb_k)$ if and only if there exist $\widetilde{\tau}_m^+,\widetilde{\tau}_m^-\geq0$ such that $T_m(\widetilde{\tau}_m)=\widetilde{\tau}_m^+-\widetilde{\tau}_m^-$ for each $m$ and $\sum_m\Tr[T_m(\widetilde{\tau}_m)]\leq1$.
    
    ($\Rightarrow$) Suppose $\tau\in\mc{T}$. Each $T_m(\widetilde{\tau}_m)$ has a Jordan--Hahn decomposition given by
    \begin{align}
        T_m(\widetilde{\tau}_m)&=\widetilde{\tau}_m^+-\widetilde{\tau}_m^-,
    \end{align}
    where $\widetilde{\tau}_m^+,\widetilde{\tau}_m^-\geq0$ and $\widetilde{\tau}_m^+$ and $\widetilde{\tau}_m^-$ are supported on orthogonal subspaces. Then, 
    \begin{align}
        \sum_m\Tr[\widetilde{\tau}_m^++\widetilde{\tau}_m^-]&= \sum_m\norm{T_m(\widetilde{\tau}_m)}_1 \notag\\
        &\leq1.
    \end{align}

    ($\Leftarrow$) Suppose there exist $\widetilde{\tau}_m^+,\widetilde{\tau}_m^-\geq0$ such that $T_m(\widetilde{\tau}_m)=\widetilde{\tau}_m^+-\widetilde{\tau}_m^-$ and $\sum_m\Tr[\widetilde{\tau}_m^++\widetilde{\tau}_m^-]\leq1$. Then, using the triangle inequality, 
    \begin{align}
        \sum_m\norm{T_m(\widetilde{\tau}_m)}_1&=\sum_m\norm{\widetilde{\tau}_m^+-\widetilde{\tau}_m^-}_1 \notag\\
        &\leq\sum_m\norm{\widetilde{\tau}_m^+}_1+\norm{\widetilde{\tau}_m^-}_1 \notag\\
        &=\sum_m \Tr[\widetilde{\tau}_m^++\widetilde{\tau}_m^-] \notag\\
        &\leq1,
    \end{align}
    concluding the proof.
\end{proof}

\section{Code for GMRE}
\label{app:multi_rains_code}

Here, we include a snippet of code from the MATLAB file attached in the ancillary files of our paper's arXiv posting. This code snippet shows how we calculate the GMRE, but it does not contain the error handling and comments included in the MATLAB file. It should be noted that our code is not written in terms of semi-definite constraints due to our usage of CVXQUAD~\cite{cvxquad}. We also use~\cite{cvx} and~\cite{qetlab}.
\begin{lstlisting}[language=Matlab]
function R = multi_rains_entanglement(rho, dims)
    rho_dims = size(rho);
    rho_row = rho_dims(1);
    n = length(dims);
    sys_vec = 1:n;
    m = n-1;
    num_partitions = 2^m-1;

    cvx_begin sdp quiet
        variable X(rho_row,rho_row,num_partitions) complex semidefinite
        
        sdp_var = 0;
        sdp_con = 0;
        for i = 1:num_partitions
            sdp_var = sdp_var + X(:,:,i);
        end
        count = 0;
        for i = 1:m
            num = nchoosek(m, i);
            sys_combos = nchoosek(sys_vec, i);
            for j = 1:num
                sdp_con = sdp_con + TraceNorm(PartialTranspose(X(:,:,j+count),sys_combos(j,1:end),dims));
            end
            count = count+num;
        end

        minimize quantum_rel_entr(rho, sdp_var)
        subject to
            sdp_con<=1
        cvx_end   
    R = cvx_optval/log(2)
    disp(sdp_var)
end
\end{lstlisting}

\section{Alternate definition of multipartite Rains entanglement}

In this appendix, we define an alternate multipartite Rains entanglement measure, which we plan to explore further in future work.

\begin{definition}
Let $\rho$ be a $k$-partite quantum state. Define
\begin{align}
\boldsymbol{R}(\rho) \coloneqq  \min_{\sigma \in \mathcal{T}'} \boldsymbol{D}(\rho \| \sigma),
\end{align}
where
\begin{align}
\mathcal{T}' 
\coloneqq  \Bigl\{ \sigma \ge 0  :  \left\| T_X(\sigma) \right\|_1 \le 1 \ \ \forall X \in \mathcal{B} \Bigr\}.
\end{align}
Here $\mathcal{B}$ denotes the set of all bipartitions of the $k$ parties, and $\boldsymbol{D}(\rho \| \sigma)$ is a generalized divergence.
\end{definition}

This quantity is an upper bound on the GMRE.
It is not clear whether this quantity is a selective PPT monotone, which represents a potential advantage of the GMRE over this quantity. However, it is a single-letter, strong converse upper bound on the GHZ distillable entanglement $E^\ep_{\operatorname{D}}(\rho)$ defined in \eqref{eq:def-1-shot-GHZ-distill-ent}, because, in addition to other properties, it obeys a subadditivity inequality, in contrast to the GMRE. It is also efficiently computable by semi-definite programming. Detailed proofs will be provided in forthcoming work. 

\end{document}